\documentclass{article}
\RequirePackage{amsfonts}                   
\RequirePackage{amssymb}                   
\RequirePackage{amsmath}
\RequirePackage{mathrsfs}	
\usepackage{stmaryrd}		
\usepackage{eucal,bm,graphicx,rotating}
\usepackage{natbib}

\ifx\theorem\undefined
\newtheorem{theorem}{Theorem}
\newtheorem{lemma}[theorem]{Lemma}

\newtheorem{corollary}[theorem]{Corollary}

\fi

\ifx\BlackBox\undefined
\newcommand{\BlackBox}{\rule{1.5ex}{1.5ex}}  
\fi

\ifx\proof\undefined
\newenvironment{proof}{\par\noindent{\bf Proof\ }}{\hfill\BlackBox\\[2mm]}
\fi

\newcommand{\EE}{\mathbb{E}}

\newcommand{\II}{\mathbb{I}}

\newcommand{\LL}{\mathbb{L}}

\newcommand{\RR}{\mathbb{R}}

\newcommand{\reals}{\mathbb{R}}
\newcommand{\test}[1]{\ensuremath{\llbracket #1 \rrbracket}}
\newcommand{\half}{\frac{1}{2}}
\newcommand{\thalf}{{\textstyle\half}}

\newcommand{\KL}{\mathrm{KL}}

\newcommand{\minimal}[1]{\underline{#1}}
\newcommand{\minL}{\minimal{L}}
\newcommand{\minLL}{\minimal{\LL}}

\newcommand{\abold}{\mathbf{a}}

\newcommand{\Xcal}{\mathcal{X}}

\newcommand{\nsp}{\hspace*{-3mm}}
\newcommand{\mypara}[1]{\vspace*{-1mm}\paragraph*{#1}}
\newcommand{\myparas}[1]{\vspace*{-5mm}\paragraph*{#1}}

\title{Generalised Pinsker Inequalities} 
\author{Mark D. Reid\\
{Australian National University}\\ 
{Canberra ACT 0200, Australia} \\
{Mark.Reid@anu.edu.au}\\
\\
Robert C. Williamson\\
{Australian National University and NICTA}\\ 
{Canberra ACT 0200, Australia} \\
{Bob.Williamson@anu.edu.au}
}

\begin{document}

\maketitle

\begin{abstract}
    {\sloppy
    We generalise the classical Pinsker inequality which relates variational
    divergence to Kullback-Liebler divergence in two ways: we consider {\em
    arbitrary} $f$-divergences in place of KL divergence, and we assume
    knowledge of a {\em sequence} of values of generalised variational
    divergences.  We then develop a best possible inequality for this doubly
    generalised situation. Specialising our result to the classical case
    provides a new and tight explicit bound relating KL to variational divergence
    (solving a problem posed by Vajda some 40 years ago). The solution relies
    on exploiting a connection between divergences  and  the
    Bayes risk of a learning problem via an integral representation.
    }
\end{abstract}

\section{Introduction}

Divergences such as the Kullback-Liebler and variational divergence arise
pervasively. They are a means of defining a notion of distance between two
probability distributions. The question often arises: given knowledge of one,
what can be said of the other? 
For all distributions $P$ and $Q$ on an arbitrary set, the classical Pinsker 
inequality relates the Kullback-Liebler divergence $KL(P,Q)$ and
variational divergence $V(P,Q)$ by $\mathrm{KL(P,Q)}\ge \half[V(P,Q)]^2$.
This simple classical bound is known not to be tight.
Over the past several decades a number of refinements have been given
(see Appendix~\ref{section:history} for a summary of past work). 

Vajda~\cite{Vajda1970} posed the question of determining a \emph{tight lower bound} on
KL-divergence in terms of variational divergence. This ``best possible Pinsker
inequality'' takes the form
\begin{equation}
L(V) :=\inf_{V(P,Q)=V} \mathrm{KL}(P,Q), \ \ \ \ V\in[0,2).
\label{eq:vajda-tight-lower-bound}
\end{equation}
Recently Fedotov et al.~\cite{FedotovTopsoe2003}
presented an {\em implicit} parametric solution of the form of the graph
of the bound as $(V(t),L(t))_{t\in\reals^+}$ where
\begin{eqnarray}
    V(t) &=& t\left(1-\left(\coth(t)-\textstyle\frac{1}{t}\right)^2\right),
     \label{eq:implicit-pinsker}\\
     L(t) &=& \log\left(\frac{t}{\sinh(t)}\right)+t\coth(t)
    -\frac{t^2}{\sinh^2(t)}\nonumber.
\end{eqnarray}

One can generalise the notion of a Pinsker inequality in at least two ways: 1)
replace KL divergence by a general $f$-divergence; and 2) bound the
$f$-divergence in terms of the known values of a {\em sequence} of generalised
variational divergences (defined later in this paper) $(V_{\pi_i})_{i=1}^n$,
$\pi_i\in(0,1)$. In this paper we study this doubly generalised problem and
provide a complete solution in terms of explicit, best possible bounds.

The main result is given below as Theorem~\ref{theorem:general-pinsker}. 
Applying it to specific $f$-divergences gives the following corollary\footnote{
	The terms $\test{V<1}$ and $\test{V\ge 1}$ are indicator functions and are
	defined below.
}.
\begin{corollary}\label{cor:special-cases}
	Let $V = V(P,Q)$ denote the variational divergence between the distributions
	$P$ and $Q$ and similarly for the other divergences in 
	Table~\ref{table:symmetric-divergences} below. 
	Then the following bounds for the divergences hold and are tight:
    \begin{equation*}
	h^2 \ge 2 - \sqrt{4-V^2};\ \ 
	\mathrm{J} \ge \textstyle 2V\ln \left(\frac{2+V}{2-V}\right);\ \ 
	\Psi \ge \frac{8V^2}{4-V^2}
	\end{equation*}
	\begin{eqnarray}
	\mathrm{I} &\ge& \left(\frac{1}{2}\!-\!\frac{V}{4}\right)\ln(2\!-\!V) \!+\!\left(\frac{1}{2}\!+\!\frac{V}{4}\right)\ln(2\!+\!V)\!-\!\ln(2)\nonumber\\
	\mathrm{T}&\!\!\!\!\!\ge\!\!\!\! &\textstyle
	   \ln\left(\frac{4}{\sqrt{4-V^2}}\right)-\ln(2) \nonumber \\
	\chi^2 &\!\!\!\!\!\ge\!\!\!\! & \textstyle\test{V<1} V^2 + \test{V\ge
	    1}\frac{V}{(2-V)}\label{eq:my-chi-squared-bound}\\
	\mathrm{KL} &\!\!\!\!\ge &\!\! \min_{\beta\in[V-2,2-V]}\textstyle 
	   \left(\frac{V+2-\beta}{4}\right)\ln\left(\frac{\beta-2-V}{\beta-2+V}\right)
    +\nonumber\\
    & &\nsp\left(\textstyle\frac{\beta+2-V}{4}\right)\ln\left(\textstyle\frac{\beta+2-V}{\beta+2+V}\right).
    \label{eq:my-pinsker-KL}
	\end{eqnarray}
\end{corollary}

The proof of the main result depends in an essential way on a learning theory
perspective. We make use of an integral representation of $f$-divergences in
terms of DeGroot's statistical information---the
difference between a prior and posterior Bayes risk\cite{DeGroot1962}. 
By using the
relationships between the generalised variational divergence and the
0-1 misclassification loss we
are able to use an elementary but somewhat intricate geometrical argument to
obtain the result.

The rest of the paper is organised as follows. 
Section~\ref{section:background} collects background results upon which we
rely. The main result of the paper is stated in Section~\ref{section:main} and
its proof
presented in in Section 
\ref{section:pinsker-proofs}. Appendix~\ref{section:history} summarises
previous work.

\section{Background Results and Notation}
\label{section:background}
In this section we collect notation and background concepts and results we need
for the main result.

\subsection{Notational Conventions}

The substantive objects are defined within the body of the paper. Here we 
collect
elementary notation and the conventions we adopt throughout.
We write $x\wedge y:=\min(x,y)$, $x\vee y :=\max(x,y)$ and 
$\test{p}=1$ if $p$ is true and $\test{p}=0$ otherwise.
The generalised function $\delta(\cdot)$ is defined by $\int_a^b \delta(x)
f(x)dx=f(0)$ when $f$ is continuous at $0$ and $a<0<b$. For convenience,
we will define $\delta_c(x) := \delta(x-c)$.
The real numbers are denoted $\reals$, the non-negative reals $\reals^+$;
Sets are in calligraphic font: $\mathcal{X}$.
Vectors are written in bold font: $\abold,\bm{\alpha},\bm{x}\in\reals^m$. 
We will often have cause to take expectations ($\EE$) over random variables. 
We write such quantities in blackboard bold:
$\II$, $\LL$, \emph{etc.} 
The lower bound on quantities with an intrinsic lower bound (e.g. the Bayes 
optimal loss) are written with an underbar: $\minL$, $\minLL$. 
Quantities related by
double integration recur in this paper and we notate the starting point in
lower case, the first integral with upper case, and the second integral in
upper case with an overbar: $\gamma$, $\Gamma$, $\bar{\Gamma}$. 

\subsection{Csisz\'{a}r $f$-divergences}
\label{sub:fdivergence}

The class of \emph{$f$-divergences} \citep{AliSilvey1966,Csiszar1967} 
provide a rich set of relations that can be used to measure the separation of 
the distributions.
An $f$-divergence is a function that measures the
``distance'' between a pair of distributions $P$ and $Q$ defined over a space
$\mathcal{X}$ of observations. 
Traditionally, the $f$-divergence of $P$ from $Q$ is defined for any 
convex $f:(0,\infty)\to\RR$ such that $f(1) = 0$. In this case, the 
$f$-divergence is
\begin{equation}\label{eq:fdiv1}
	\II_f(P,Q) = \EE_Q \left[ f\left(\frac{dP}{dQ}\right) \right]
			   = \int_{\Xcal} f\left(\frac{dP}{dQ}\right)\,dQ
\end{equation}
when $P$ is absolutely continuous with respect to $Q$ and equal $\infty$ 
otherwise.\footnote{
	Liese and Miescke~\cite[pg. 34]{Liese2008} give a definition that does not require 
	absolute continuity. 
}

All $f$-divergences are non-negative and zero when $P=Q$, that is, 
$\II_f(P,Q) \ge 0$ and $\II_f(P,P) = 0$ for all distributions $P, Q$. 
In general, however, they are not metrics, since they are not necessarily 
\emph{symmetric} (\emph{i.e.}, for all distributions
$P$ and $Q$, $\II_{f}(P,Q) = \II_{f}(Q,P)$) and do not necessarily satisfy
the triangle inequality.

Several well-known divergences correspond to specific choices of the function 
$f$ \cite[\S 5]{AliSilvey1966}.
One divergence central to this paper is the \emph{variational divergence} 
$V(P,Q)$ which is obtained by setting $f(t) = |t-1|$ in Equation~\ref{eq:fdiv1}.
It is the only $f$-divergence that is a true metric on the space of 
distributions over $\Xcal$ \citep{Khosravifard:2007} and gets its name from its 
equivalent definition in the variational form
\begin{equation}
	\label{eq:fdivvar}
	V(P,Q) = 2 \| P - Q \|_{\infty} 
	:= 2 \sup_{A \subseteq \Xcal} | P(A) - Q(A) |.
\end{equation}
(Some authors define $V$ without the 2 above.)
Furthermore, the variational divergence is one of a family of ``primitive'' 
or ``simple'' $f$-divergences discussed in Section~\ref{sub:intrep}. 
These are primitive in the sense that all other $f$-divergences can be expressed 
as a weighted sum of members from this family. 
 
Another well known $f$-divergence is the Kullback-Leibler (KL) divergence 
$\mathrm{KL}(P,Q)$, obtained by setting $f(t) = t \ln(t)$ in 
Equation~\ref{eq:fdiv1}. Others are given in 
Table~\ref{table:symmetric-divergences}.

As already mentioned in the introduction, the KL and variational divergences
satisfy the classical Pinsker's inequality which states that for all 
distributions $P$ and $Q$ over some common space $\Xcal$
\begin{equation}
    \mathrm{KL}(P,Q) \ge \thalf [V(P,Q)]^2.
\end{equation}

\subsection{Integral Representations of $f$-divergences}\label{sub:intrep}
The main tool in our proof of Theorem~\ref{theorem:general-pinsker} is the 
integral representation of $f$-divergences, first articulated by 
\"{O}sterreicher and Vajda \cite{OsterreicherVajda1993a} and 
Gutenbrunner \cite{Gutenbrunner:1990}. 
They show that an $f$-divergence can be represented as a weighted integral of 
the ``simple'' divergence measures 
\begin{equation}
	V_\pi(P,Q) = \II_{f_\pi}(P,Q),
\end{equation}
where $f_{\pi}(t) := \min\{\pi, 1-\pi\} - \min\{1-\pi, \pi t\}$ for 
$\pi\in[0,1]$.

\begin{theorem}\label{thm:f-intrep}
	For any convex $f$ such that $f(1)=0$, the $f$-divergence $\II_f$ 
	can be expressed, for all distributions $P$ and $Q$, as
	\begin{equation}\label{eq:LV-integral-rep}
		\II_f(P,Q) = \int_0^1 V_{\pi}(P,Q)\,\gamma_f(\pi)\,d\pi
	\end{equation}
	where the (generalised) function
	\begin{equation}\label{eq:gamma-f-relationship}
		\gamma_f(\pi) := \frac{1}{\pi^3} f''\left(\frac{1-\pi}{\pi}\right).
	\end{equation}
\end{theorem}
Recently, this theorem has been shown to be a direct consequence
of a generalised Taylor's expansion for convex functions 
\cite{LieVaj06,ReidWilliamson2009}.

Even when $f$ is not twice differentiable, the convexity of $f$ implies
its continuity and so its right-hand derivative $f'_{+}$ exists. In this case, 
$\gamma$ is interpreted distributionally in terms of $df'_{+}$.
For example, when $f(t) = |t-1|$ then $f''(t) = 2\delta(t-1)$ and so 
$\gamma_f(\pi) = 2\frac{1}{\pi^3}\delta(1-2\pi) = 16\delta_{\half}(\pi)$.

The divergences $V_\pi$ for $\pi\in[0,1]$ can be seen as a family of generalised 
variational divergences since, $df'_{+}(t)$ for any member of this family is 
$\pi\delta(t-\frac{1-\pi}{\pi})$ and so
$\gamma_{f_\pi} = \frac{1}{\pi^2}\delta_{\frac{1-\pi}{\pi}}$.
Thus, for $\pi=\thalf$ we have $\gamma_{f_{\half}} = 4\delta_{\half}$,
that is, four times the 
$\gamma$ function for variational divergence and so by 
(\ref{eq:LV-integral-rep}) we see that 
\begin{equation}\label{eq:V-Vpi}
	V(P,Q) = 4V_{\half}(P,Q).
\end{equation}

Theorem~\ref{thm:f-intrep} shows that knowledge of the values of $V_\pi(P,Q)$ 
for all $\pi\in[0,1]$ is sufficient to compute the value of $\II_f(P,Q)$ for any
$f$-divergence, since the weight function $\gamma$ is dependent only on $f$, 
not $P$ and $Q$. All of the generalised Pinsker bounds we derive are found by
asking how knowledge of a the value of a finite number of $V_\pi(P,Q)$ 
constrains the overall value of $\II_f(P,Q)$.

Table~\ref{table:symmetric-divergences} summarises the weight functions $\gamma$ 
for a number of $f$-divergences that appear in the literature. 
These are used in the proof of specific bounds in 
Corollary~\ref{cor:special-cases}.

\begin{sidewaystable*}
    \begin{center}
    \renewcommand{\arraystretch}{1.3}
    \begin{tabular*}{14cm}{llll}
	\hline
	Symbol 
	& Divergence Name
	& $f(t)$ 
	& $\gamma(\pi)$ 
	\\
	\hline\hline
	$\rule{0pt}{2.4ex} V(P,Q)$ 
	& Variational 
	& $|t-1|$
	& $16\delta\left(\pi-\frac{1}{2}\right)$ 
	\\[0.2mm]
	\hline
	$\mathrm{KL}(P,Q) $ 
	& Kullback-Liebler 
	& $t\ln t$
	& $ \frac{1}{\pi^2 (1-\pi)}$ 
	\\[0.2mm]
	\hline
	$\rule{0pt}{2.4ex}\Delta(P,Q)$ 
	& Triangular Discrimination 
	& $(t-1)^2/(t+1)$
	& 8 
	\\
	\hline
	$\rule{0pt}{2.4ex}\mathrm{I}(P,Q)$	
	& Jensen-Shannon 
	& $\frac{t}{2}\ln t - \frac{(t+1)}{2}\ln(t+1) +\ln 2$ 
	& $\frac{1}{2\pi(1-\pi)}$ 
	\\[0.2mm]
	\hline
	$\rule{0pt}{3ex} \mathrm{T}(P,Q)$ 
	& Arithmetic-Geometric Mean 
	& $\left(\frac{t+1}{2}\right)\ln\left(\frac{t+1}{2\sqrt{t}}\right)$
	& $\frac{(2\pi-\frac{1}{2})^{2}+\frac{1}{2}}{4\pi^2(\pi-1)^2}$ 
	\\[0.2mm]
	\hline
	$\rule{0pt}{2.4ex}\mathrm{J}(P,Q)$ 
	& Jeffreys 
	& $(t-1)\ln(t)$
	& $\frac{1}{\pi^2 (1-\pi)^2}$ 
	\\[0.2mm]
	\hline
	$\rule{0pt}{2.4ex}h^2(P,Q)$ 
	& Hellinger 
	& $(\sqrt{t}-1)^2$
	& $\frac{1}{2[\pi(1-\pi)]^{3/2}}$ 
	\\[0.2mm]
	\hline
	$\rule{0pt}{3ex} \chi^2(P,Q)$ 
	& Pearson $\chi$-squared 
	& $(t-1)^2$
	& $\frac{2}{\pi^3}$ 
	\\[0.2mm]
	\hline
	$\rule{0pt}{3ex}\Psi(P,Q)$ 
	& Symmetric $\chi$-squared 
	& $\frac{(t-1)^2(t+1)}{t}$ 
	& $\frac{2}{\pi^3}+\frac{2}{(1-\pi)^3}$ 
	\\[0.2mm]
	\hline
    \end{tabular*}
    \caption{
     Divergences and their corresponding functions $f$ and weights $\gamma$;
    confer~\cite{Taneja2005,LieVaj06}. 
    Tops{\o}e~\cite{Topsoe:2000} calls $C(P,Q)=2 \mathrm{I}(P,Q)$ and
    $\tilde{C}(P,Q)=2\mathrm{T}(P,Q)$ the Capacitory and Dual Capacitory
    discrimination respectively. Several of the above divergences are
	``symmetrised'' versions of others. For example, 
	$T(P,Q)=\frac{1}{2}[\mathrm{KL}(\frac{P+Q}{2},P) +\mathrm{KL}(\frac{P+Q}{2},Q)]$, $I(P,Q)=\frac{1}{2}[\mathrm{KL}(P,\frac{P+Q}{2})+\mathrm{KL}(Q,\frac{P+Q}{2})]$,
	$\mathrm{J}(P,Q)=\mathrm{KL}(P,Q)+\mathrm{KL}(Q,P)$, and
	$\Psi(P,Q)=\chi^2(P,Q)+\chi^2(Q,P)$.
    \label{table:symmetric-divergences}}
\end{center}
\end{sidewaystable*}

Before we can prove the main result, we need to establish some properties of
the general variational divergences. In particular, we will make use of 
their relationship to Bayes risks for 0-1 loss.

\subsection{Divergence and Risk}
Let $\minLL(\pi,P,Q)$ denote the 0-1 Bayes risk for a classification problem
in which observations are drawn from $\Xcal$ using the mixture distribution 
$M = \pi P + (1-\pi)Q$, and each observation $x\in\Xcal$ is assigned a 
positive label with probability $\eta(x) := \pi\frac{dP}{dM}(x)$. 
If $r = r(x) \in\{0,1\}$ is a label prediction for a particular 
$x\in\Xcal$, the 0-1 expected loss for that observation is 
\[
	L(r,\pi, p, q) 
	= (1-\pi)q\test{r = 1} + \pi p\test{r = 0}.
\]
where $q = \frac{dQ}{dM}(x)$ and $p = \frac{dP}{dM}(x)$ are densities.
Thus, 
the full expected 0-1 loss of a predictor $r :\Xcal \to \{0,1\}$ is given by 
$\LL(r, \pi, P, Q) := \EE_M[L(r(x),\pi, p(x), q(x))]$
and it is well known (\emph{e.g.}, \cite{DevGyoLug96}) that its Bayes risk is 
obtained by the Bayes optimal predictor 
$r^*(x) := \test{\eta(x) \ge \half}$. That is,
\begin{equation}
	\minLL(\pi,P,Q) 
	:= \inf_{r} \LL(r,\pi,P,Q) 
	= \LL(r^*, \pi, P, Q), \label{eq:bayesrisk}
\end{equation}
where the infimum is taken over all ($M$-measurable) predictors 
$r:\Xcal\to\{0,1\}$.
So, by the definition of $\eta(x)$ and noting that $\eta \ge \half$ iff 
$\pi p \ge \half (\pi p + (1-\pi)q)$ which holds iff
$\pi p \ge (1-\pi)q$ we see that
the 0-1 Bayes risk can be expressed as
\begin{eqnarray}
	\lefteqn{\minLL(\pi,P,Q)} \\
	\nsp&\nsp=& \EE_M [ 
			(1-\pi)q \test{\eta  \ge \thalf} +\ \pi p \test{\eta < \thalf} 
		] \nonumber \\
	\nsp&\nsp=& (1-\pi)\EE_Q[\test{\pi p \ge (1-\pi)q}] 
	    + \pi\EE_P[\test{\pi p < (1-\pi)q }]. \nonumber
\end{eqnarray}

We now observe that 
\[
	q f_\pi\left(\frac{p}{q}\right) 
	= ((1-\pi)\wedge\pi) q - 
	\begin{cases}
		(1-\pi)q, & q(1-\pi) \le \pi p \\
		\pi p,    & q(1-\pi) > \pi p
	\end{cases}
\]
and so by noting that 
$\EE_Q \left[ f_\pi\left(\frac{dP}{dQ}\right)\right] 
= \EE_M \left[ q f_\pi\left(\frac{p}{q}\right) \right]$ we have established the 
following lemma.
\begin{lemma}\label{lem:vandrisk}
	For all $\pi\in[0,1]$ and all distributions $P$ and $Q$, the generalised
	variational divergence satisfies
	\begin{eqnarray}
		V_\pi(P,Q)
		&=& (1-\pi)\wedge\pi - \minLL(\pi,P,Q) \label{eq:statinfo}.
	\end{eqnarray}	
\end{lemma}
Thus, the value of $V_\pi(P,Q)$ can be understood via the 0-1 Bayes risk for a
classification problem with label-conditional distributions $P$ and $Q$ and 
prior probability $\pi$ for the positive class.
This relationship between $f$-divergence and Bayes risk is not new. 
It was established in a more general setting by \"{O}sterreicher and Vajda 
\cite{OsterreicherVajda1993a} (who note that the term in (\ref{eq:statinfo}) is 
the statistical information for 0-1 loss) and later by Nguyen 
\emph{et al.} \cite{Nguyen:2005}.

\subsection{Concavity of 0-1 Bayes Risk Curves}\label{sub:concavity}

For a given pair of distributions $P$ and $Q$ the set of values for 
$\minLL(\pi,P,Q)$ as $\pi$ varies over $[0,1]$ can be visualised as a curve
as in Figure~\ref{figure:bounds}.

\begin{lemma}\label{lem:concavity}
	For all distributions $P$ and $Q$, 
	the function $\pi \mapsto \minLL(\pi, P, Q)$ is concave.
\end{lemma}
\begin{proof}
	By (\ref{eq:bayesrisk}) we have that 
	\[
	\minLL(\pi,P,Q) = \EE_M[L(r^*, \pi, p, q)].
	\]
	Observe that
	\begin{eqnarray*}
		L(r^*, \pi, p, q) 
		&=& (1-\pi)q\test{\eta \ge \thalf} + \pi p \test{\eta < \thalf} \\
		&=& \begin{cases}
			(1-\pi)q, & q(1-\pi) \le \pi p \\
			\pi p,    & q(1-\pi) > \pi p
		\end{cases} \\
		&=& \min\{(1-\pi)q, \pi p \}
	\end{eqnarray*}
	and so for any $p, q$ is the minimum of two linear functions and thus
	concave in $\pi$.
	The full Bayes risk is the expectation of these functions and thus simply a 	
	linear combination of concave functions and thus concave.
\end{proof}

The tightness of the bounds in the main result of the next section depend
on the following corollary of a result due to Torgersen \cite{torgersen1991cse}.
It asserts that any appropriate concave function can be viewed as the 0-1 risk 
curve for \emph{some} pair of distributions $P$ and $Q$.
A proof can be found in \cite[\S 6.3]{ReidWilliamson2009}.
\begin{corollary}
    Suppose $\Xcal$ has a connected component. 
    Let $\psi\colon [0,1]\rightarrow[0,1]$ be an arbitrary concave function
    such that  for all $\pi\in[0,1]$, $0\le\psi(\pi)\le \pi\wedge(1-\pi)$. Then
    there exists $P$ and $Q$ such that $\minLL(\pi,P,Q)=\psi(\pi)$ for all
    $\pi\in[0,1]$.
    \label{corollary:all-bayes-risk-curves-possible}
\end{corollary}

\section{Main Result}
\label{section:main}
We will now show how viewing $f$-divergences in terms of their weighted
integral representation simplifies the problem of understanding the
relationship between different divergences and leads, amongst other things, to
an explicit formula for~(\ref{eq:vajda-tight-lower-bound}).

Fix a positive integer $n$. 
Consider a sequence $0<\pi_1 <\pi_2 < \cdots <\pi_n < 1$. Suppose we
``sampled'' the value of $V_\pi(P,Q)$ at these discrete values of $\pi$.
Since $\pi\mapsto V_{\pi}(P,Q)$ is concave, the piece-wise linear  
concave function passing through points 
\[
\{(\pi_i, V_{\pi_i}(P,Q))\}_{i=1}^n
\]
is guaranteed to be an upper bound on the
variational curve $(\pi, V_{\pi}(P,Q))_{\pi\in(0,1)}$.  
This therefore
gives a lower bound on the $f$-divergence given
by a weight function $\gamma$. 
This observation forms the basis of the theorem stated below.

\begin{theorem}
    \label{theorem:general-pinsker}
For a  positive integer $n$
consider a sequence $0<\pi_1<\pi_2 <\cdots <\pi_n< 1$. Let $\pi_0:=0$ and
$\pi_{n+1}:=1$ and for $i=0,\ldots,n+1$ let
\[
	\psi_i:=(1-\pi_i)\wedge\pi_i - V_{\pi_i}(P,Q)
\]
(observe that consequently
$\psi_0=\psi_{n+1}=0$). Let 
\begin{eqnarray}
    A_n &:=&\left\{\abold=(a_1,\ldots,a_n)\in\reals^n
    \colon \phantom{ \frac{\psi_{i+1}-\psi_i}{\pi_{i+1}-\pi_i}
    }\right.\label{eq:A-n}\\
   & & \left. \frac{\psi_{i+1}-\psi_i}{\pi_{i+1}-\pi_i} \le
   a_i \le\frac{\psi_i-\psi_{i-1}}{\pi_i-\pi_{i-1}},\ i=1,\ldots,n\right\}
   .\nonumber 
\end{eqnarray}
The set $A_n$ defines the allowable slopes of a piecewise linear function
majorizing $\pi\mapsto V_{\pi}(P,Q)$ at each of $\pi_1,\ldots,\pi_n$.
For $\abold=(a_1,\ldots,a_n)\in A_n$, let 
\begin{eqnarray}
\nsp\tilde{\pi}_i \nsp&:= &\nsp\frac{\psi_i\!-\!\psi_{i+1}\!+\!a_{i+1}
    \pi_{i+1}\!-\!a_i\pi_i}{a_{i+1}-a_i},\   i=0,\ldots,n,\\
\nsp j \nsp&:= &\nsp \{k\in\{1,\ldots,n\}: \tilde{\pi}_{k} <\textstyle\frac{1}{2}\le 
	\tilde{\pi}_{k+1}\}.  \label{eq:j-def}\\
\nsp \bar{\pi}_i\nsp &:=&\nsp\test{i<j}\tilde{\pi}_i +
    \test{i=j} \textstyle\frac{1}{2} + \test{j<i}\tilde{\pi}_{i-1},
      \label{eq:bar-pi}\\
\nsp \alpha_{\abold,i}\nsp &:= &\nsp \test{i\le
j}(1-a_i)+\test{i>{j}}(-1-a_{i-1}),\label{eq:alpha-abold}\\
\nsp \beta_{\abold,i}\nsp&:=&\nsp\test{i\!\le\!{j}}(\psi_i\!-\!a_i\pi_i)\!+\!
\test{i\!>\!{j}}(\psi_{i-1}\!\!-\!a_{i-1}\pi_{i-1})\label{eq:beta-abold}
\end{eqnarray}
for $i=0,\ldots,n+1$ and let $\gamma_f$ be the weight corresponding to  
$f$ given by
(\ref{eq:gamma-f-relationship}).

For arbitrary $\II_f$ and 
    for all distributions $P$ and $Q$ the following bound holds.
    If in addition $\mathcal{X}$ contains a connected component, it is tight.
    \begin{eqnarray}
	\nsp&  & \II_f(P,Q) \\
	\nsp&\ge & \min_{\abold\in A_n} \sum_{i=0}^{n}
	\int_{\bar{\pi}_i}^{\bar{\pi}_{i+1}} 
	(\alpha_{\abold,i}\pi +\beta_{\abold,i}) \gamma_f(\pi) d\pi 
	\label{eq:general-pinsker}\\
	\nsp&=& \min_{\abold\in A_n} \sum_{i=0}^{n} \left[
	    \left(\alpha_{\abold,i}\bar{\pi}_{i+1}+
	    \beta_{\abold,i}\right)\Gamma_f(\bar{\pi}_{i+1})
	    -\alpha_{\abold,i}\bar{\Gamma}_f(\bar{\pi}_{i+1}) \right.\nonumber\\
	\nsp&& \ \ \ \ \left. -\left(\alpha_{\abold,i}\bar{\pi}_{i} +
	    \beta_{\abold,i}\right)\Gamma_f(\bar{\pi}_i)
	+\alpha_{\abold,i}\bar{\Gamma}_f(\bar{\pi}_i)\right] ,
	\label{eq:general-pinsker-evaluated}
\end{eqnarray}
where $\Gamma_f(\pi):=\int^\pi \gamma_f(t) dt$ and 
$\bar{\Gamma}_f(\pi):=\int^\pi \Gamma_f(t) dt$.
\end{theorem}
Equation \ref{eq:general-pinsker-evaluated} follows from
(\ref{eq:general-pinsker}) by integration by parts. The remainder of the 
proof is in Section~\ref{section:pinsker-proofs}. 
Although (\ref{eq:general-pinsker-evaluated}) looks daunting, we observe: 
(1) the
constraints on $\abold$ are convex (in fact they are a box constraint); and (2)
the objective is a relatively benign function of $\abold$.  

When $n=1$ the result simplifies considerably. If in addition $\pi_1=\thalf$  
then by (\ref{eq:V-Vpi}) we have $V_{\frac{1}{2}}(P,Q)=\frac{1}{4}V(P,Q)$. 
It is then a straightforward exercise to explicitly evaluate
(\ref{eq:general-pinsker}), especially when $\gamma_f$ is symmetric. The
following theorem expresses the result in terms of $V(P,Q)$ for comparability
with previous results. The result for $\mathrm{KL}(P,Q)$ is a
(best-possible) improvement on the classical Pinsker
inequality. 
\begin{theorem}\label{theorem:special-cases-pinsker}
    For any distributions $P,Q$ on $\mathcal{X}$, let $V:=V(P,Q)$. 
    Then the following bounds hold and, if in 
    addition $\mathcal{X}$ has a connected component, are tight.

    When $\gamma$ is symmetric about $\thalf$ and convex,
    \begin{equation}
	\II_f(P,Q) \ge 2 \left[ 
	\bar{\Gamma}_f\left(\textstyle\frac{1}{2}-\textstyle\frac{V}{4}\right)+
	\textstyle\frac{V}{4}\Gamma_f\left(\textstyle\frac{1}{2}\right)-
	\bar{\Gamma}_f\left(\textstyle\frac{1}{2}\right)\right]
	\label{eq:general-symmetric-pinsker}
    \end{equation}
    and $\Gamma_f$ and $\bar{\Gamma}_f$ are as in 
    Theorem~\ref{theorem:general-pinsker}.
\end{theorem}
This theorem gives the first explicit representation of 
the optimal Pinsker bound.\footnote{
     A summary of existing results and their relationship to those presented
     here is given in appendix \ref{section:history}.
}
By plotting both (\ref{eq:implicit-pinsker}) and
(\ref{eq:my-pinsker-KL}) one can confirm 
that  the two bounds (implicit and explicit) coincide; see Figure
\ref{figure:pinsker-curves}.
\begin{figure}[t]
    \begin{center}
	\includegraphics[width=0.5\textwidth]{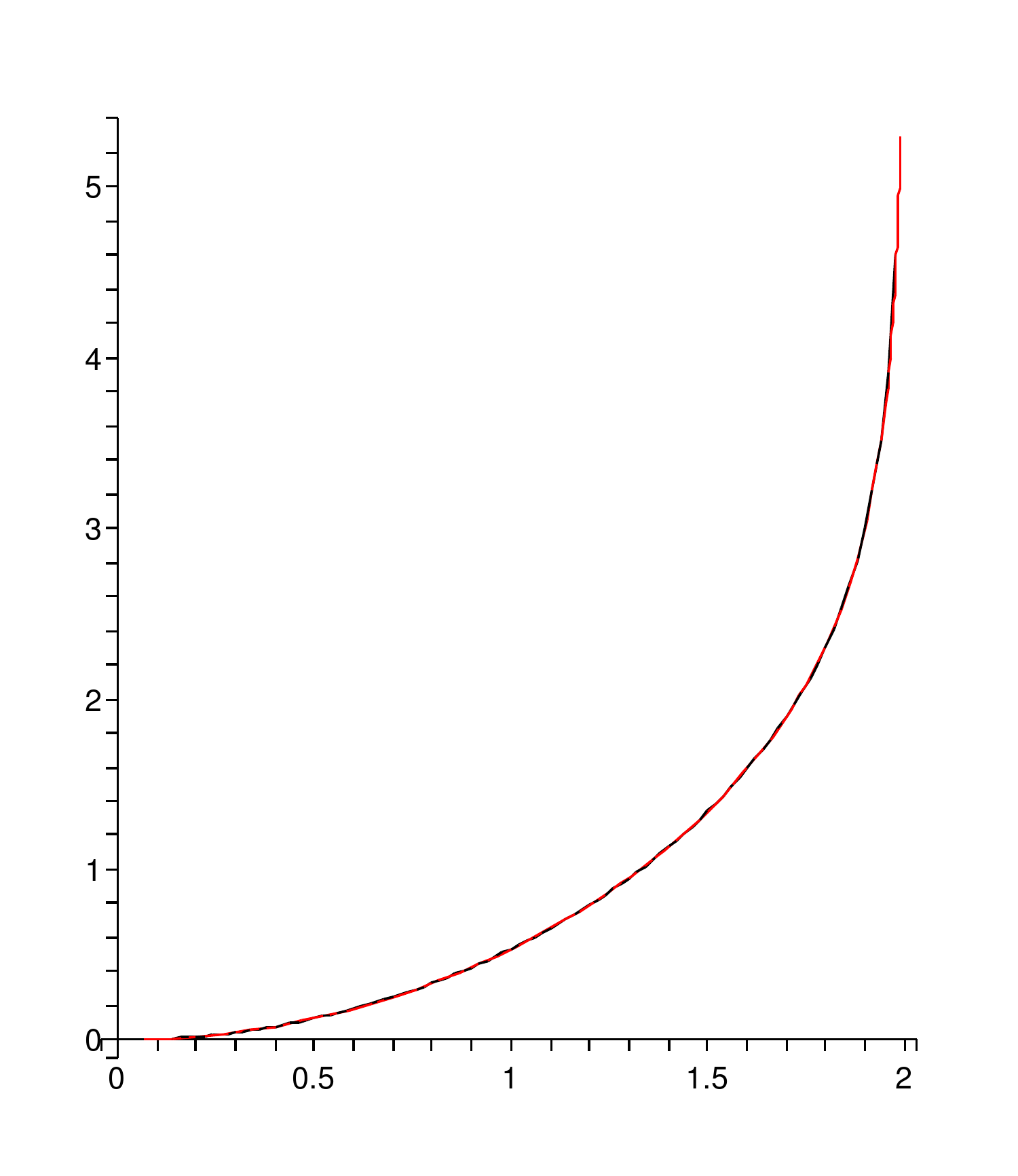}
    \end{center}
\caption{Lower bound on  $\mathrm{KL}(P,Q)$ 
as a function of the variational divergence $V(P,Q)$. Both the explicit bound
(\ref{eq:my-pinsker-KL}) and Fedotorev et al.'s implicit bound
(\ref{eq:implicit-pinsker}) are plotted.
\label{figure:pinsker-curves}}
\end{figure}

\section{Proof of Main Result}
\label{section:pinsker-proofs}
\begin{proof} {\bf (Theorem \ref{theorem:general-pinsker})}
	This proof is driven by the duality between the family of variational
	divergences $V_{\pi}(P,Q)$ and the 0-1 Bayes risk $\minLL(\pi,P,Q)$
	given in Lemma~\ref{lem:vandrisk}.
    Given distributions $P$ and $Q$ let 
    \begin{equation*}
		\phi(\pi)
		= V_{\pi}(P,Q)
		= \pi\wedge(1-\pi)-\psi(\pi),
    \end{equation*}
    where $\psi(\pi)=\minLL(\pi,P,Q)$. We
    know that $\psi$ is non-negative and  concave and satisfies 
    $\psi(\pi)\le\pi\wedge(1-\pi)$
    and thus $\psi(0)=\psi(1)=0$.

    Since 
    \begin{equation}
		\II_f(P,Q) = \int_0^1 \phi(\pi) \gamma_f(\pi) d\pi,
		\label{eq:I-f-gamma}
    \end{equation}
    $\II_f(P,Q)$ is minimised by minimising $\phi$ over all
    $(P,Q)$ such that
    \[
	\phi(\pi_i)=\phi_i=\pi_i\wedge(1-\pi_i)-\psi(\pi_i).
    \]
    Since 
    $\psi_i:= (1-\pi_i)\wedge\pi_i-V_{\pi_i}(P,Q) = \psi(\pi_i)$
    the minimisation problem for $\phi$ can be expressed in terms of
    $\psi$ as: 
    \begin{eqnarray}
	\nsp\nsp\mbox{Given\ } (\pi_i,\psi_i)_{i=1}^n\nsp
	& & \nsp\mbox{find the maximal\ }
    \psi\!\colon\![0,1]\rightarrow[0,\thalf]
    \label{eq:S-problem}\\
     \mbox{such that\ }& &\nsp\nsp \psi(\pi_i) = \psi_i, \ \ 
           i=0,\ldots,n+1,\label{eq:psi-i-constraint}\\
	& &\nsp\nsp  \psi(\pi) \le \pi\wedge (1-\pi), \ \
	\pi\in[0,1],\label{eq:psi-overbound-constraint}\\
	& & \nsp\nsp\psi  \mbox{\ is concave}.\label{eq:psi-concave-constraint}
    \end{eqnarray}
    This will tell us the optimal $\phi$ to use since optimising over $\psi$ is
    equivalent to optimising over $\minLL(\cdot,P,Q)$. Under the additional
    assumption on $\mathcal{X}$, Corollary
    \ref{corollary:all-bayes-risk-curves-possible} implies that for any $\psi$
    satisfying (\ref{eq:psi-i-constraint}),
    (\ref{eq:psi-overbound-constraint})  and (\ref{eq:psi-concave-constraint})
    there exists $P,Q$ such that $\minLL(\cdot,P,Q)=\psi(\cdot)$.
	This establishes the tightness of our bounds.

    Let $\Psi$ be the set of piece-wise linear concave functions on $[0,1]$
    having $n+1$ segments such that
    $\psi\in\Psi\Rightarrow\psi$ satisfies 
    (\ref{eq:psi-i-constraint}) and (\ref{eq:psi-overbound-constraint}). 
    We now show
    that in order to solve (\ref{eq:S-problem}) it suffices to consider
    $\psi\in\Psi$.

    \begin{figure*}[t]
	\begin{center}
	\includegraphics[width=0.98\textwidth]{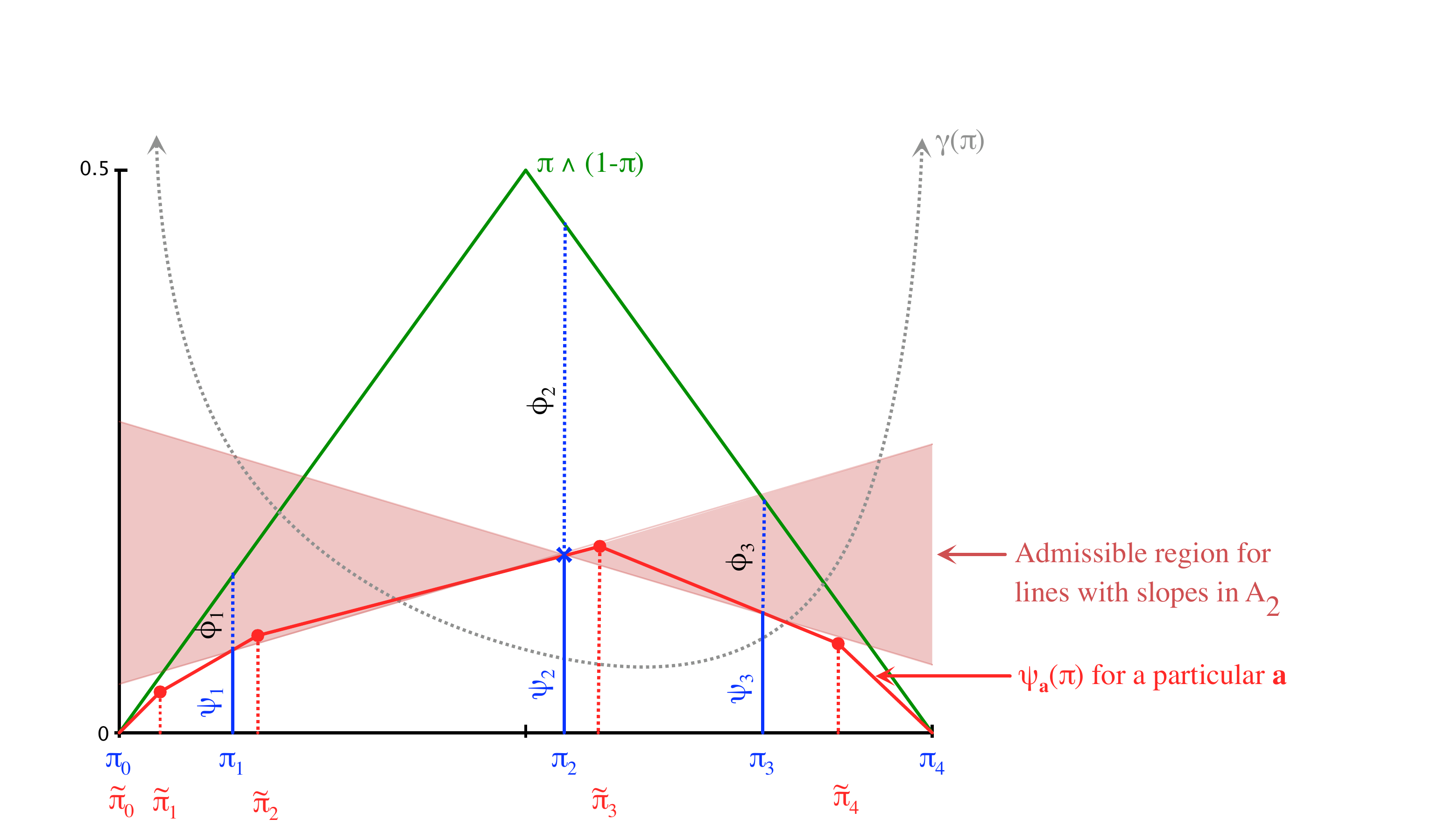}
    \end{center}
	\caption{Illustration of construction of optimal
	$\psi(\pi)=\minLL(\pi,P,Q)$ in the proof of 
	Theorem~\ref{theorem:general-pinsker}. 
	The optimal $\psi$ is piece-wise linear
	such that $\psi(\pi_i)=\psi_i$, $i=0,\ldots,n+1$. \label{figure:bounds}}
    \end{figure*}

    If $g$ is a concave function on $\reals$, then let
    \[
    \eth g(x):= \{s\in\reals\colon g(y) \le g(x) + \langle s, y-x\rangle, \
    y\in\reals\}
    \]
    denote the {\em sup-differential} of $g$ at $x$. (This is 
    the obvious analogue of the \emph{sub}-differential for convex functions
    \citep{Rockafellar:1970}.) 
    Suppose $\tilde{\psi}$ is a general concave function satisfying
    (\ref{eq:psi-i-constraint}) and (\ref{eq:psi-overbound-constraint}).
    For $i=1,\ldots,n$, let
    \begin{eqnarray*}
    & &\nsp\nsp G_i^{\tilde{\psi}}:=\left\{
    [0,1]\ni g_i^{\tilde{\psi}}:\pi_i\mapsto \psi_i\in\reals  \ \ 
    \mbox{is linear and}\right. \\
    & & \ \ \ \ \ \ \  
    \textstyle\left.\left.\frac{\partial}{\partial\pi}g_i^{\tilde{\psi}}(\pi)
    \right|_{\pi=\pi_i} \in\eth\tilde{\psi}(\pi_i)\right\}.
    \end{eqnarray*}
    Observe that by concavity, for all concave $\tilde{\psi}$ satisfying 
    (\ref{eq:psi-i-constraint}) and (\ref{eq:psi-overbound-constraint}), for
    all $g\in\bigcup_{i=1}^n G_i^{\tilde{\psi}}$, $g(\pi)\ge\psi(\pi)$,
    for $\pi\in[0,1]$. 
    
    Thus given any such $\tilde{\psi}$, one can always
    construct
    \begin{equation}
	\psi^*(\pi)=\min(g_1^{\tilde{\psi}}
	(\pi),\ldots,g_n^{\tilde{\psi}}(\pi))
	\label{eq:psi-general}
    \end{equation}
    such that $\psi^*$ is concave,  satisfies
    (\ref{eq:psi-i-constraint})   and
    $\psi^*(\pi)\ge\tilde{\psi}(\pi)$, for all $\pi\in[0,1]$. It remains to
    take account of (\ref{eq:psi-overbound-constraint}). That is trivially done
    by setting
    \begin{equation}
	 \psi(\pi) = \min (\psi^*(\pi),\pi\wedge(1-\pi))
	 \label{eq:psi-general-2}
    \end{equation}
    which remains concave and piecewise linear (although with potentially one
    additional linear segment). Finally, the pointwise smallest concave 
    $\psi$
    satisfying 
    (\ref{eq:psi-i-constraint}) and (\ref{eq:psi-overbound-constraint}) is the
    piecewise linear function connecting the points 
    $(0,0), (\pi_1,\psi_1),(\pi_2,\psi_2),\ldots, (\pi_m,\psi_m), (1,0)$.

    Let $g\colon [0,1]\rightarrow[0,\frac{1}{2}]$ be this function which can be written
    explicitly as
    \[
    g(\pi) =\left(\psi_i +\frac{(\psi_{i+1}-\psi)(\pi-\pi_i)}{\pi_{i+1}-\pi_i}
    \right)\cdot \test{\pi\in[\pi_i,\pi_{i+1}]},
    \]
    where we have defined $\pi_0:=0$, $\psi_0:=0$, $\pi_{n+1}:=1$ and
    $\psi_{n+1}:=0$. 
    
    We now explicitly parametrize  this family of functions.
    Let $p_i\colon[0,1]\rightarrow\reals$ denote the
    affine segment the graph of which passes through 
    $(\pi_i,\psi_i)$, $i=0,\ldots,n+1$.
    Write $p_i(\pi)=a_i\pi +b_i$. We know that $p_i(\pi_i)=\psi_i$ and thus
    \begin{equation}
	 b_i=\psi_i -a_i\pi_i, \ \ \ i=0,\ldots,n+1 .
    \end{equation}
    In order to determine the constraints on $a_i$, since $g$ is concave and
    minorizes $\psi$, it suffices to only consider $(\pi_{i-1}, g(\pi_{i-1}))$
    and $(\pi_{i+1},g(\pi_{i+1}))$ for $i=1,\ldots,n$.  We have (for
    $i=1,\ldots,n$)
    \begin{eqnarray}
	  & p_i(\pi_{i-1}) &\ge g(\pi_{i-1})\nonumber\\
	 \Rightarrow & a_i \pi_{i-1}+b_i &\ge \psi_{i-1}\nonumber\\
	 \Rightarrow & a_i \pi_{i-1}+\psi_i-a_i\pi_i &\ge \psi_{i-1}\nonumber\\
	 \Rightarrow & a_i \underbrace{(\pi_{i-1}-\pi_i)}_{< 0}& \ge \psi_{i-1}
	 -\psi_i\nonumber\\[-3mm]
	 \Rightarrow & a_i &\le \frac{\psi_{i-1}-\psi_i}{\pi_{i-1}-\pi_i} .
	 \label{eq:ai-upperbound}
    \end{eqnarray}
    Similarly we have (for $i=1,\ldots,n$)
    \begin{eqnarray}
	 & p_i(\pi_{i+1}) &\ge g(\pi_{i+1})\nonumber\\
	 \Rightarrow & a_i\pi_{i+1}+b_i &\ge\psi_{i+1}\nonumber\\
	 \Rightarrow & a_i\pi_{i+1}+\psi_i-a_i\pi_i &\ge \psi_{i+1}\nonumber\\
	 \Rightarrow & a_i\underbrace{(\pi_{i+1}-\pi_i)}_{>0} &\ge
	 \psi_{i+1}-\psi_i\nonumber\\[-3mm]
	 \Rightarrow & a_i & \ge
	 \frac{\psi_{i+1}-\psi_i}{\pi_{i+1}-\pi_i} .\label{eq:ai-lowerbound}
    \end{eqnarray}
    We now determine the points at which $\psi$ defined by
    (\ref{eq:psi-general}) and (\ref{eq:psi-general-2}) change slope. That
    occurs at the points $\pi$ when
    \begin{eqnarray*}
      & p_i(\pi) &=p_{i+1}(\pi)\\
      \Rightarrow & a_i\pi+\psi_i-a_i\pi_i
      &=a_{i+1}\pi+\psi_{i+1}-a_{i+1}\pi_{i+1}\\
      \Rightarrow & (a_{i+1}-a_i)\pi &=
      \psi_i-\psi_{i+1}+a_{i+1}\pi_{i+1}-a_i\pi_i\\
      \Rightarrow &\pi &=
      \frac{\psi_i-\psi_{i+1}+a_{i+1}\pi_{i+1}}{a_{i+1}-a_i} \\
      &  & =:\tilde{\pi}_i
\end{eqnarray*}
    for $i=0,\ldots,n$. Thus
    \[
    \psi(\pi)=p_i(\pi),\ \ \ \pi\in[\tilde{\pi}_{i-1},\tilde{\pi}_i], \
    i=1,\ldots,n.
    \]
    Let $\abold=(a_1,\ldots,a_n)$. We explicitly denote the dependence of
    $\psi$ on $\abold$ by writing $\psi_{\abold}$.
    Let
    \begin{eqnarray*}
    \phi_{\abold}(\pi)\nsp&:=&\nsp\pi\wedge(1-\pi) -\psi_{\abold}(\pi)\\
    &=&\nsp\alpha_{\abold,i}\pi+\beta_{\abold,i}, \ 
    \pi\in[\bar{\pi}_{i-1},\bar{\pi}_{i}],  i=1,\ldots,n+1,
    \end{eqnarray*}
    where $\abold\in A_n$ (see (\ref{eq:A-n})), $\bar{\pi}_i$, 
    $\alpha_{\abold,i}$ and
    $\beta_{\abold,i}$ are defined by (\ref{eq:bar-pi}), 
    (\ref{eq:alpha-abold}) and
    (\ref{eq:beta-abold}) respectively. The extra segment induced at index $j$
    (see (\ref{eq:j-def})) is needed since $\pi\mapsto\pi\wedge(1-\pi)$ has a
    slope change at $\pi=\thalf$. Thus in general, $\phi_{\abold}$ is
    piece-wise linear with $n+2$ segments (recall $i$ ranges from $0$ to $n+2$);
    if $\tilde{\pi}_{k+1}=\frac{1}{2}$ for some $k\in\{1,\ldots,n\}$, then there will
    be only $n+1$ non-trivial segments.

    Thus  
    \[
    \left\{\pi\mapsto\sum_{i=0}^{n} \phi_{\abold}(\pi)
    \cdot\test{\pi\in[\bar{\pi}_{i},
    \bar{\pi}_{i+1}]}\colon \abold\in A_n\right\}
    \]
    is the set of $\phi$ consistent with the constraints and $A_n$ is defined
    in (\ref{eq:A-n}).
    Thus substituting into (\ref{eq:I-f-gamma}), interchanging the order of
    summation and integration and optimizing we have
    shown~(\ref{eq:general-pinsker}). The tightness has already been argued: 
    under the additional assumption on $\mathcal{X}$, since there is no
    slop in the argument above since every $\psi$ satisfying the
    constraints in (\ref{eq:S-problem}) is the Bayes risk function for some 
	$(P,Q)$.
\end{proof}

\begin{proof} {\bf (Theorem \ref{theorem:special-cases-pinsker})}
    In this case $n=1$  and the optimal $\psi$ function will be piecewise
    linear, concave, and its graph will 
    pass through $(\pi_1,\psi_1)$. Thus the optimal $\phi$
    will be of the form
    \[
    \phi(\pi) = \left\{
    \begin{array}{ll}
	 0, & \pi\in [0,L]\cup [U,1]\\
	 \pi -(a\pi+b), \ & \pi\in [L,\thalf]\\
	 (1-\pi)-(a\pi+b), & \pi\in[\thalf,U].
    \end{array}
    \right.
    \]
    where $a\pi_1+b=\psi_1\Rightarrow b=\psi_1-a\pi_1$ and $a\in
    [-2\psi_1,2\psi_1]$ (see Figure \ref{figure:SimplePinsker}).
    \begin{figure}[t]
	\begin{center}
	\includegraphics[width=0.8\textwidth]{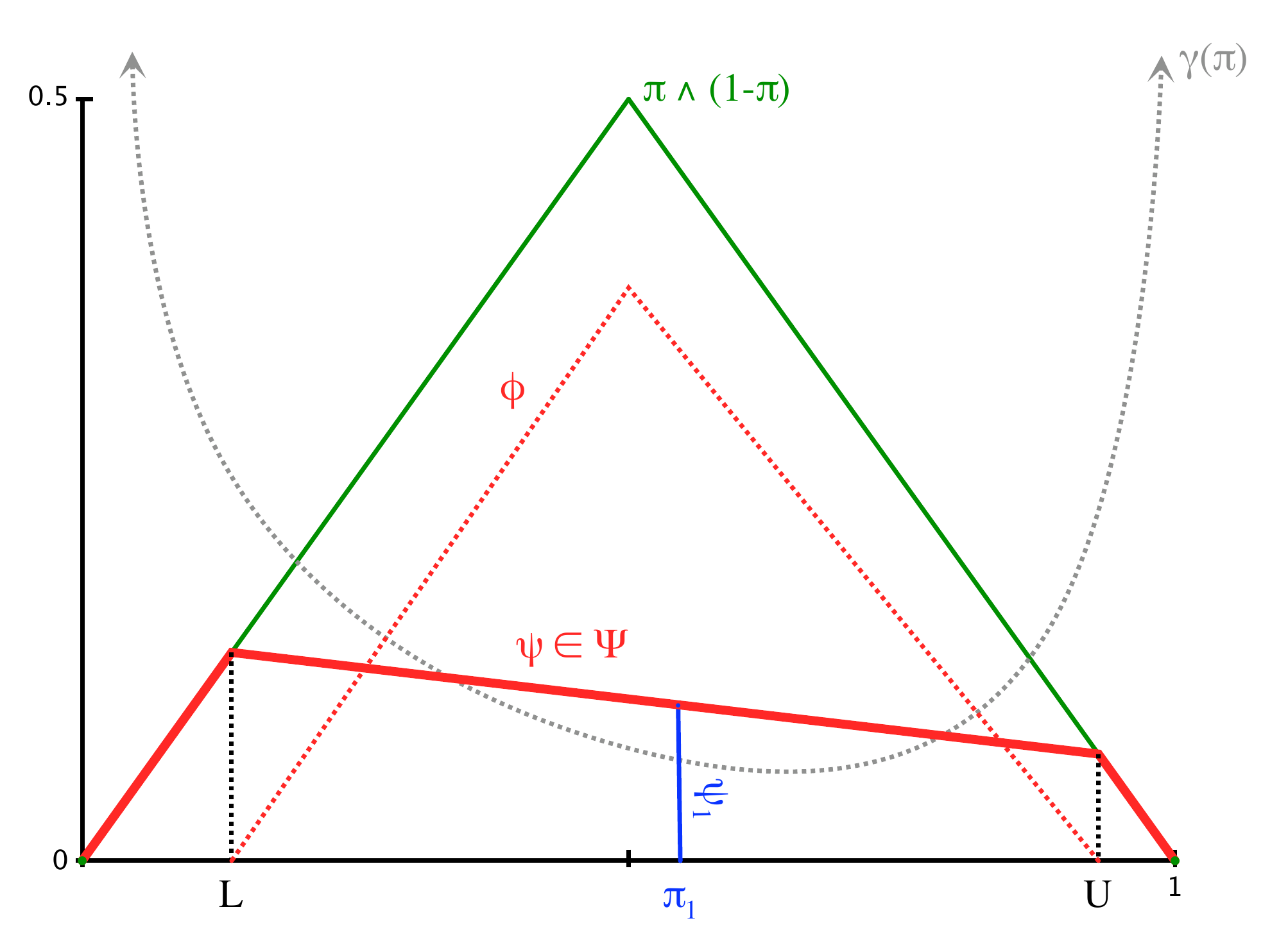}
	\end{center}
	\caption{The optimisation problem when $n=1$. Given $\psi_1$, there are
	many risk curves consistent with it. The optimisation problem involves
	finding the piece-wise linear concave risk curve $\psi\in\Psi$ and the 
	corresponding  $\phi = \pi\wedge(1-\pi) - \psi$ that 
	maximises $\II_f$. $L$ and $U$ are defined in the text.
	\label{figure:SimplePinsker}}
    \end{figure}
    For variational divergence, $\pi_1=\thalf$ and thus by (\ref{eq:V-Vpi})
    \begin{equation}
	\psi_1=\pi_1\wedge(1-\pi_1) - \frac{V}{4}=\frac{1}{2}-\frac{V}{4}
	\label{eq:psi-1-variational}
    \end{equation}
    and so $\phi_1=V/4$. We can thus determine $L$ and $U$: 
    \begin{eqnarray*}
	 & aL +b & = L\\
	 \Rightarrow & aL+\psi_1-a\pi_1& =L\\
	 \Rightarrow  & L & = \frac{a\pi_1 -\psi_1}{a-1}.
    \end{eqnarray*}
    Similarly $ aU+b =1-U \Rightarrow U =\frac{1-\psi_1+a\pi_1}{a+1}$
    and thus
    \begin{eqnarray}
	\nsp& &\nsp \II_f(P,Q) \ge\!\!\! \min_{a\in[-2\psi_1,2\psi_1]}\!\!\!
	\int\limits_{\frac{a\pi_1-\psi_1}{a-1}}^{\frac{1}{2}}\!\!\!\!\!
	[(1-a)\pi-\psi_1+a\pi_1]
	\gamma_f(\pi) d\pi \nonumber\\
	\nsp& &\nsp\ \ \ \ +\int\limits_{\frac{1}{2}}^{\frac{1-\psi_1+a\pi_1}{a+1}}\!\!\!\!\!
	[(-a-1)\pi-\psi_1+a\pi_1+1]\gamma_f(\pi)d\pi.
	\label{eq:general-I-f-in-terms-of-psi-1}
    \end{eqnarray}
If $\gamma_f$ is symmetric about $\pi=\thalf$ and convex and
$\pi_1=\thalf$, then the optimal $a=0$.  Thus in that case, 
\begin{eqnarray}
    \II_f(P,Q) &\ge& 2 \int_{\psi_1}^{\frac{1}{2}} 
      (\pi-\psi_1) \gamma_f(\pi) d\pi \\
      &=& 2\left[(\thalf-\psi_1)\Gamma_f(\thalf) + \bar{\Gamma}_f(\psi_1)
      -\bar{\Gamma}_f(\thalf)\right]\nonumber\\
      &=& 2 \left[\textstyle\frac{V}{4} \Gamma_f(\thalf)
      +\bar{\Gamma}_f \left(\thalf-\textstyle\frac{V}{4}\right)
      -\bar{\Gamma}_f(\thalf)\right] . \label{eq:simpler-pinsker-Gamma}
\end{eqnarray}
Combining the above with (\ref{eq:psi-1-variational}) leads to a range of
Pinsker style bounds for symmetric $\II_f$:

    \mypara{Jeffrey's Divergence} Since $J(P,Q)=\KL(P,Q)+\KL(Q,P)$ we have
	$\gamma(\pi)=\frac{1}{\pi^2(1-\pi)^2}$. (As a check,
	$f(t)=(t-1)\ln(t)$, $f''(t)=\frac{t+1}{t^2}$ and so
	$\gamma_f(\pi)=\frac{1}{\pi^3}f''\left(\frac{1-\pi}{\pi}\right)=
	\frac{1}{\pi^2 (1-\pi)^2}$.)
	Thus
	\begin{eqnarray*}
	    J(P,Q) &\ge & 2 \int_{\psi_1}^{1/2} \frac{(\pi-\psi_1)}{\pi^2
	    (1-\pi)^2} d\pi\\
	    &=& (4\psi_1 -2)(\ln(\psi_1)-\ln(1-\psi_1)).
	\end{eqnarray*}
	Substituting $\psi_1=\frac{1}{2}-\frac{V}{4}$ gives
	\[
	    J(P,Q) \ge 
	    V\ln\left(\frac{2+V}{2-V}\right).
	\]

	Observe that the above bound behaves like $V^2$ for small $V$, and 
	$V\ln\left(\frac{2+V}{2-V}\right)\ge V^2$ for $V\in[0,2]$.
	Using
	the traditional Pinkser inequality ($\KL(P,Q)\ge V^2/2$) we have
	\begin{eqnarray*}
	    J(P,Q)&=& \KL(P,Q)+\KL(Q,P)\\
	    &  \ge & \frac{V^2}{2}+\frac{V^2}{2}
	    = V^2.
	\end{eqnarray*}
\myparas{Jensen-Shannon Divergence} Here
	$f(t)=\frac{t}{2}\ln t - 
	     \frac{(t+1)}{2}\ln(t+1) +\ln 2$ 
	 and thus
	$\gamma_f(\pi)=\frac{1}{\pi^3}
	f''\left(\frac{1-\pi}{\pi}\right)=\frac{1}{2\pi(1-\pi)}$. Thus
	\begin{eqnarray*}
	    I(P,Q)\nsp &=&\nsp 2 \int_{\psi_1}^{\frac{1}{2}}
	    \frac{\pi-\psi_1}{2\pi(1-\pi)}d\pi\\
	    \nsp&=&\nsp \ln(1\!-\!\psi_1)-\psi_1\ln(1\!-\!\psi_1)+\psi_1\ln\psi_1+\ln(2).
	\end{eqnarray*}
	Substituting $\psi_1=\frac{1}{2}-\frac{V}{4}$ leads to
	\[
	I(P,Q) \ge \left(\textstyle\frac{1}{2}-\textstyle\frac{V}{4}\right)\ln(2-V) +
	\left(\textstyle\frac{1}{2}+\textstyle\frac{V}{4}\right)\ln(2+V) -\ln(2).
	\]
\myparas{Hellinger Divergence} Here $f(t)=(\sqrt{t}-1)^2$. Consequently
    $\gamma_f(\pi) = \frac{1}{\pi^3} f''\left(\frac{1-\pi}{\pi}\right)=
    \frac{1}{\pi^3} \frac{1}{2 \left( (1-\pi)/
    \pi\right)^{3/2}}=\frac{1}{2[\pi(1-\pi)]^{3/2}}$
    and thus
    \begin{eqnarray*}
	h^2(P,Q) &\ge & 2\int_{\psi_1}^{\frac{1}{2}}
	\frac{\pi-\psi_1}{2[\pi(1-\pi)]^{3/2}} d\pi\\
	&=& \frac{4 \sqrt{\psi_1} (\psi_1-1)
	+2\sqrt{1-\psi_1}}{\sqrt{1-\psi_1}}\\
	&=&
	\frac{4\sqrt{\frac{1}{2}-\frac{V}{4}}\left(\frac{1}{2}-\frac{V}{4}-1\right)+2\sqrt{1-\frac{1}{2}+\frac{V}{4}}}{\sqrt{1-\frac{1}{2}+\frac{V}{4}}}\\
	&=&2- \frac{(2+V)\sqrt{2-V}}{\sqrt{2+V}}\\
	&=& 2-\sqrt{4-V^2}.
    \end{eqnarray*}
    For small $V$, $2-\sqrt{4-V^2}\approx V^2/4$.
\mypara{Arithmetic-Geometric Mean Divergence} In this case, $f(t)=\frac{t+1}{2}
    \ln\left(\frac{t+1}{2\sqrt{t}}\right)$. Thus $f''(t)=\frac{t^2+1}{4t^2
    (t+1)}$ and hence $\gamma_f(\pi)=\frac{1}{\pi^3}
    f''\left(\frac{1-\pi}{\pi}\right) = 
    \gamma_f(\pi)=\frac{2\pi^2-2\pi+1}{\pi^2(\pi-1)^2}$
    and thus
    \begin{eqnarray*}
	T(P,Q) &\ge & 2\int_{\psi_1}^{\frac{1}{2}} (\pi-\psi_1)
    \frac{2\pi^2-2\pi+1}{\pi^2(\pi-1)^2} d\pi\\
    &=& -\frac{1}{2}\ln(1-\psi)-\frac{1}{2}\ln(\psi)-\ln(2).
    \end{eqnarray*}
    Substituting $\psi_1=\frac{1}{2}-\frac{V}{4}$ gives
    \begin{eqnarray*}
    T(P,Q) &\ge &  -\frac{1}{2}\ln\left(\frac{1}{2}+\frac{V}{4}\right)
    -\frac{1}{2}\ln\left(\frac{1}{2}-\frac{V}{4}\right)-\ln(2)\\
    & = & \ln\left(\frac{4}{\sqrt{4-V^2}}\right)-\ln(2) .
\end{eqnarray*}
\mypara{Symmetric $\chi^2$-Divergence} In this case $\Psi(P,Q)=\chi^2(P,Q)+\chi^2(Q,P)$
    and thus (see below) $\gamma_f(\pi)=\frac{2}{\pi^3}+\frac{2}{(1-\pi)^3}$.
    (As a check, from $f(t)=\frac{(t-1)^2(t+1)}{t}$ we have
    $f''(t)=\frac{2(t^3+1)}{t^3}$ and thus
    $\gamma_f(\pi)=\frac{1}{\pi^3}f''\left(\frac{1-\pi}{\pi}\right)$ gives the
    same result.)
    \begin{eqnarray*}
	\Psi(P,Q) &\ge& 2\int_{\psi_1}^{\frac{1}{2}}
	(\pi-\psi_1)\left(\frac{2}{\pi^3}+\frac{2}{(1-\pi)^3}\right) d\pi\\
	&=& \frac{2(1+4\psi_1^2-4\psi_1)}{\psi_1(\psi_1-1)}.
    \end{eqnarray*}
    Substituting $\psi_1=\frac{1}{2}-\frac{V}{4}$ gives $\Psi(P,Q)\ge
    \frac{8V^2}{4-V^2}$.

When $\gamma_f$ is not symmetric, one needs to use
(\ref{eq:general-I-f-in-terms-of-psi-1}) instead of the simpler
(\ref{eq:simpler-pinsker-Gamma}).  We consider two  cases.
\mypara{$\chi^2$-Divergence} Here $f(t)=(t-1)^2$ and so $f''(t)=2$ and hence
    $\gamma(\pi)=f''\left(\frac{1-\pi}{\pi}\right)/\pi^3 = \frac{2}{\pi^3}$ which is not
    symmetric.
	Upon substituting $2/\pi^3$ for $\gamma(\pi)$ in
	(\ref{eq:general-I-f-in-terms-of-psi-1}) and evaluating the integrals
	we obtain
	\[
	\chi^2(P,Q) \ge 2 \min_{a\in[-2\psi_1,2\psi_1]} 
	\underbrace{\textstyle\frac{1+4\psi_1^2-4\psi_1}{2\psi_1-a} -
	\frac{1+4\psi_1^2-4\psi_1}{2\psi_1-a-2}}_{=:J(a,\psi_1)} .
	\]
	One can then solve $\frac{\partial}{\partial a} J(a,\psi_1)=0$ for $a$
	and one obtains $a^*=2\psi_1-1$. Now $a^*> -2\psi_1$ only if
	$\psi_1>\frac{1}{4}$. One can check that when $\psi_1\le\frac{1}{4}$,
	then $a\mapsto J(a,\psi_1)$ is monotonically increasing for
	$a\in[-2\psi_1,2\psi_1]$ and hence the minimum occurs at
	$a^{*}=-2\psi_1$. Thus the value of $a$ minimising
	$J(a,\psi_1)$ is
	\[
	a^*= \test{\psi_1>1/4} (2\psi_1-1) + \test{\psi_1\le 1/4}(-2\psi_1).
	\]
	Substituting the optimal value of $a^*$ into $J(a,\psi_1)$ we obtain
	\begin{eqnarray*}
	    J(a^*,\psi_1)&=& \test{\psi_1>1/4}(2+8\psi_1^2-8\psi_1)\\
	    \nsp& &\nsp\nsp\nsp\nsp + 
	    \test{\psi_1\le 1/4}\left(\frac{1+4\psi_1^2-4\psi}{4\psi} -
	    \frac{1+4\psi_1^2-4\psi}{4\psi_1-2}\right) .
	\end{eqnarray*}
	Substituting $\psi_1=\frac{1}{2}-\frac{V}{4}$ and observing that $V<1
	\Rightarrow \psi_1>1/4$ we obtain
	\[
	\chi^2(P,Q) \ge \test{V<1} V^2 + \test{V\ge 1}\frac{V}{(2-V)}.
	\]
	Observe that the bound diverges to $\infty$ as $V\rightarrow 2$.

\mypara{Kullback-Leibler Divergence} In this case 
    we have $f(t)=t\ln t$ and thus $f''(t)=1/t$ and consequently 
    $\gamma_f(\pi)=\frac{1}{\pi^3}f''\left(\frac{1-\pi}{\pi}\right)=
    \frac{1}{\pi^2(1-\pi)}$ which is clearly not symmetric. From  
	(\ref{eq:general-I-f-in-terms-of-psi-1}) we obtain
	\begin{eqnarray*}
	\nsp && \nsp \KL(P,Q) \ge \min_{[-2\psi_1,2\psi_1]}
	\left(1-\textstyle\frac{a}{2}-\psi_1\right)
	\ln\left(\textstyle\frac{a+2\psi_1-2}{a-2\psi_1}\right) \\
	 &&\ \ \ \ \ \ \ 
	+\left(\textstyle\frac{a}{2}+\psi_1\right)\ln\left(\textstyle\frac{a+2\psi_1}{a-2\psi_1+2}\right).
    \end{eqnarray*}
	Substituting $\psi_1=\frac{1}{2}-\frac{V}{4}$ gives
	\[
	    \KL(P,Q) \ge  \min_{a\in\left[\frac{V-2}{2},\frac{2-V}{2}\right]}
	    \delta_a (V),
	\]
	where 
	\[\delta_a(V)\!=\!\textstyle\left(\frac{V+2-2a}{4}\right)
	\ln\left(\frac{2a-2-V}{2a-2+V}\right)+
	\left(\frac{2a+2-V}{4}\right)\ln\left(\frac{2a+2-V}{2a+2+V}\right). 
	\]
	Set $\beta:=2a$ and we have (\ref{eq:my-pinsker-KL}).
\end{proof}

\section{Conclusion}
We have generalised the classical Pinsker inequality and developed best
possible bounds for the general situation. A special case of the result gives
an explicit bound relating Kullback-Liebler divergence and variational
divergence. The proof relied on an integral representation of $f$-divergences
in terms of statistical information. Such representations are a powerful
device as they identify the primitives underpinning general learning problems.
These representations are further studied in~\cite{ReidWilliamson2009}.

\appendix
\section{History of Pinsker Inequalities}
\label{section:history}
    Pinsker \cite{Pinsker1964} presented the first bound relating $\mathrm{KL}(P,Q)$ to
    $V(P,Q)$: $\mathrm{KL}\ge V^2/2$ and it is now known by his name  or 
    sometimes as the Pinsker-Csisz\'{a}r-Kullback
    inequality since Csiszar~\cite{Csiszar1967} presented another version
    and 
    Kullback \cite{Kullback1967} showed $\mathrm{KL}\ge V^2/2 +V^4/36$. Much later 
    Tops{\o}e \cite{Topsoe2001} showed $\mathrm{KL}\ge V^2/2 + V^4/36+V^6/270$.
    Non-polynomial bounds are due to Vajda~\cite{Vajda1970}: $\mathrm{KL}\ge
    L_{\mathrm{Vajda}}(V):=\log\left(\frac{2+V}{2-V}\right)-\frac{2V}{2+V}$ and
    Toussaint~\cite{Toussaint1978} who showed $\mathrm{KL}\ge L_{\mathrm{Vajda}}(V) \vee
    (V^2/2+V^4/36+V^8/288)$. 

    Care needs to be taken when comparing results from the literature as
    different definitions for the divergences exist. For example
    Gibbs and Su~\cite{GibbsSu2002} used a definition of $V$ that differs by a factor of 
    2 from ours.
    There are some isolated bounds relating $V$ to some other divergences,
    analogous to the classical Pinkser bound; Kumar~\cite{Kumar2005} has presented
    a summary as well as new bounds for a wide range of \emph{symmetric}
    $f$-divergences by making assumptions on the likelihood ratio: $r\le
    p(x)/q(x)\le R<\infty$ for all $x\in\mathcal{X}$. 
    This line of reasoning has also been
    developed by Dragomir et al.~\cite{DragomirGluvsvcevicPearce2001} and Taneja~\cite{Taneja2005, Taneja2005a}. 
    Tops{\o}e~\cite{Topsoe:2000}  has presented some infinite series  representations for
    capacitory discrimination in terms of triangular discrimination which lead
    to inequalities between those two divergences.
    Liese and Miescke~\cite[p.48]{Liese2008} give the inequality $V\le h\sqrt{4-h^2} $
    (which seems to be originally due to LeCam~\cite{LeCam1986})
    which when rearranged corresponds exactly to the bound for $h^2$ in theorem
    \ref{theorem:special-cases-pinsker}.
    Withers~\cite{Withers1999} has also presented some inequalities between other
    (particular) pairs of divergences; his reasoning is also in terms of
    infinite series expansions. 

    Arnold et al.~\cite{ArnoldMarkowichToscaniUnterreiter} considered the case of $n=1$
    but arbitrary $\II_f$ (that is they bound an arbitrary $f$-divergence in
    terms of the variational divergence). Their argument is similar to the
    geometric proof of Theorem \ref{theorem:general-pinsker}.
    They do not compute any of the explicit
    bounds in theorem \ref{theorem:special-cases-pinsker} except they state
    (page 243)  $\chi^2(P,Q)\ge V^2$ which is looser than 
    (\ref{eq:my-chi-squared-bound}).

    Gilardoni~\cite{Gilardoni:2006} showed (via an intricate argument) that if
    $f'''(1)$ exists, then $\II_f \ge \frac{f''(1) V^2}{2}$. He also showed some
    fourth order inequalities of the form $\II_f\ge c_{2,f} V^2 + c_{4,f} V^4$
    where the constants depend on the behaviour of $f$ at 1 in a complex way.
    Gilardoni~\cite{Gilardoni2006,Gilardoni2006a} presented a completely different
    approach which obtains many of the results of 
    theorem~\ref{theorem:special-cases-pinsker}.\footnote{We were unaware
    of these two papers until completing the results presented in the main
    paper. }
    Gilardoni~\cite{Gilardoni2006a} improved Vajda's bound  slightly to $
    \mathrm{KL}(P,Q) \ge \ln\frac{2}{2-V} - \frac{2-V}{2}\ln\frac{2+V}{2}$.

    Gilardoni~\cite{Gilardoni2006,Gilardoni2006a}
    presented a general tight lower bound for $\II_f=\II_f(P,Q)$ in terms of
    $V=V(P,Q)$ which is difficult to evaluate explicitly in general:
    \[
    \II_f \ge
    \frac{V}{2}\left(\frac{f[g_R^{-1}(k(1/V))]}{g_R^{-1}(k(1/V))-1} +
    \frac{f[g_L^{-1}(k(1/V))]}{1-g_L^{-1}(k(1/V))}\right),
    \]
    where $k^{-1}(t)=\frac{1}{2}\left(\frac{1}{1-g_L^{-1}(t)} +
    \frac{1}{g_R^{-1}(t)-1}\right)$ and of course $k(u)=(k^{-1})^{-1}(u)$; 
    and $g(u)=(u-1)f'(u)-f(u)$, $g_R^{-1}[g(u)]=u$
    for $u\ge 1$ and $g_L^{-1}[g(u)]=u$ for $u\le 1$. 
    He presented a new parametric form for $\II_f=\mathrm{KL}$ in terms of
    Lambert's $W$ function.
    In general, the
    result is analogous to that of
    Fedotov et al.~\cite{FedotovTopsoe2003} in that it is in a parametric form which, 
    if one wishes
    to evaluate for a particular $V$, one needs to do a one dimensional
    numerical search --- as complex as (\ref{eq:my-pinsker-KL}).
    However, when $f$ is such that
    $\II_f$ is symmetric, this simplifies to the elegant form $\II_f\ge
    \frac{2-V}{2}f\left(\frac{2+V}{2-V}\right)-f'(1)V$. 
    He presented explicit special cases for $h^2$, $J$,$\Delta$ and $I$
    identical to the results in Theorem \ref{theorem:special-cases-pinsker}.
    It  is not apparent how 
    the approach of Gilardoni~\cite{Gilardoni2006,Gilardoni2006a} could be extended to
    more general situations such as that in Theorem
    \ref{theorem:general-pinsker} (i.e.~$n>1$).

    Bolley and Villani~\cite{BolleyVillani2005}
    considered {\em weighted} versions of the Pinsker inequalities (for
    a weighted generalisation of Variational divergence) in terms of
    KL-divergence that are related to transportation inequalities.

    \subsection*{Acknowledgements}
    This work was supported by the Australian Research Council and NICTA; 
    an initiative of the Commonwealth Government under Backing Australia's
    Ability.

\bibliographystyle{plain}
\bibliography{pinsker}
\end{document}